\documentclass[11pt]{article}

\usepackage{amsmath,amsfonts,amssymb,amscd,latexsym,mathrsfs,amsthm}

\newtheorem{theorem}{Theorem}
\newtheorem{lemma}[theorem]{Lemma}
\newtheorem{proposition}[theorem]{Proposition}
\newtheorem{corollary}[theorem]{Corollary}
\newtheorem{assumption}[theorem]{Conjecture}

\newtheorem{remark}[theorem]{Remark}

\newcommand{\C}{\ensuremath{\mathbb{C}}}
\newcommand{\N}{\ensuremath{\mathbb{N}}}
\newcommand{\Z}{\ensuremath{\mathbb{Z}}}
\newcommand{\Q}{\ensuremath{\mathcal{Q}}}
\newcommand{\R}{\ensuremath{\mathbb{R}}}
\renewcommand{\S}{\mathbb{S}}
\newcommand{\bR}{\overline{\R}}
\newcommand{\bRp}{\overline{\R_+}}
\def\FF{\mathcal F}
\def\H{\mathcal H}

\def\d{\mathrm{d}}
\def\SS{\mathcal S}
\newcommand{\K}{\mathcal K}
\newcommand{\B}{\mathcal B}

\newcommand{\J}{\mathcal J}
\newcommand{\E}{\mathcal E}

\newcommand{\A}{\mathcal A}
\newcommand{\ind}{\mbox{ind}}
\newcommand{\tr}{\mathrm{tr}}
\newcommand{\Tr}{\mathrm{Tr}}
\def\F21{{}_2F_1}
\def\l{\ell}
\def\U{\mathcal U}
\def\e{\varepsilon}

\def\HH{\mathscr H}

\def\bv{\boldsymbol{\varphi}}

\begin{document}

\title{On the wave operators and Levinson's theorem \\
for potential scattering in $\R^3$}

\author{Johannes Kellendonk$\,^1$ and Serge Richard$\,^2$}
\date{\small}
\maketitle

\begin{quote}
\begin{itemize}
\item[$^1$] Universit\'e de Lyon, Universit\'e Lyon I, CNRS UMR5208, Institut Camille Jordan, 43 blvd du 11 novembre 1918, 69622 Villeurbanne Cedex, France;\\
E-mail: {\tt kellendonk@math.univ-lyon1.fr}
\item[$^2$] Graduate School of Pure and Applied Sciences,
University of Tsukuba,
1-1-1 Tennodai, Tsukuba,
Ibaraki 305-8571, Japan;\\
 E-mail: {\tt richard@math.univ-lyon1.fr}
\\[\smallskipamount]
On leave from Universit\'e de Lyon, Universit\'e Lyon I, CNRS UMR5208, Institut Camille Jordan, 43 blvd du 11 novembre 1918, 69622 Villeurbanne Cedex, France
\end{itemize}
\end{quote}


\begin{abstract}
The paper is a presentation of recent
investigations on potential scattering in $\R^3$.
We advocate a new formula for the wave operators and deduce the various outcomes that follow from this formula.
A topological version of Levinson's theorem is proposed by interpreting it as an index theorem.
\end{abstract}

\section{Introduction}

In recent works
we proposed a topological approach to Levinson's theorem by
interpreting it as an index theorem. For that purpose we introduced
a $C^*$-algebraic framework, and one of the key
ingredients of our approach is the fact that the wave operators belong
to a certain $C^*$-algebra. Once such an affiliation property is
settled, the machinery of non-commutative topology leads naturally to
the index theorem.
For various scattering systems, this program was successfully applied:
potential scattering in one dimension has
been investigated in \cite{KR5}, models of one point
interaction in $\R^n$ with $n\in \{1,2,3\}$ have been solved in
\cite{KR1}, the Friedrichs model was studied in \cite{RT}, various results on the Aharonov-Bohm model have been obtained in \cite{KPR,PR} and both \cite{KR2,KR3} are review papers describing the main idea of these investigations.
Furthermore, we have been kindly informed by its authors that a closely related work on scattering theory for lattice operators in dimension $d\geq 3$ is in preparation \cite{BSB}.
Now, the present paper is a presentation of our
investigations on potential scattering in $\R^3$.

On the way of proving the affiliation property for potential scattering in $\R^3$ we were stimulated to
conjecture an even stronger result: The shape of the wave operators is
much more rigid than what had ever been expected. More precisely,
if $W_-$ denotes one of the wave operators of the scattering system
with a potential that has a sufficiently rapid decrease to $0$ at
infinity, and if $S$ denotes the corresponding scattering
operator, then
\begin{equation}\label{newformula}
W_- = 1+ \bv(A)(S -
1) + K\ ,
\end{equation}
where $\bv(A)$ is a function of the generator of dilation that is going
to be explicitly given below, and $K$ is a compact operator in
$L^2(\R^3)$. Once this formula is obtained, the affiliation problem
becomes a simple corollary.

Let us immediately say that we do not present here the complete proof of the above
formula, but shall try to motivate it and explain the various
outcomes that follow from this formula. It should become clear to the
reader that such a result is in perfect concordance with many already
known results on scattering theory in $\R^3$, and that it provides even
a natural background for many of them.
In fact, the missing part in the proof is the compactness of the remainder term $K$. Its proof is technically difficult and has almost no relation with the material which is going to be presented here. For that reason, we have decided not to obscure these outcomes and to assume from the beginning that $K$ is compact, see Conjecture \ref{hyp}. Then, one of the main consequence of formula \eqref{newformula} is a topological version of Levinson's theorem. In the second part of this introduction we describe in non-technical terms one of its previous formulation as well as our new approach.

In the Hilbert space $\H:=L^2(\R^3)$, let us consider the operators $H_0:=-\Delta$ and $H:=-\Delta + V$ for a potential $V$ which vanishes sufficiently rapidly at infinity. In such a situation, it is known that the wave operators $W_\pm$ exist and are asymptotically complete, and that the scattering operator $S$ is unitary. We denote by $\{S(\lambda)\}_{\lambda \in \R_+}$ the scattering operator in the spectral representation of $H_0$, {\it i.e.}~$S(\lambda)$ is a unitary operator in $\HH:=L^2(\S^2)$ for almost every $\lambda$.
Then, Levinson's theorem is a relation between the number $N$ of bound states
of $H$ and an expression related to the scattering part of the
system. The latter expression can be written either in terms of an
integral over the time delay, or as an evaluation of the spectral
shift function (see the review papers \cite{Bolle,Ma} and references
therein). In particular, under suitable hypotheses on $V$ \cite{BO,Martin} a common
form of Levinson's theorem is
\begin{equation}\label{LevMartin}
\hbox{$\frac{1}{2\pi}$}\int_0^\infty \big(\tr\big[iS(\lambda)^*
\hbox{$\frac{\d S}{\d \lambda}(\lambda)$}\big]-
\hbox{$\frac{c}{\sqrt{\lambda}}$}\big) \d \lambda = N +\nu,
\end{equation}
where $\tr$ is the trace on $\HH$ and $c=(4\pi)^{-1}\int_{\R^3} \d x\;\!V(x)$.
The correction term $\nu$ arises from the existence of resonance for
$H$ at energy $0$. If such a $0$-energy resonance exists, the
correction $\nu$ is equal to $1/2$, and it is $0$ otherwise. The
explanation for the presence of $\nu$ in \eqref{LevMartin} is
sometimes quite ad hoc.

We shall now show how to rewrite \eqref{LevMartin} as an index
theorem. Our approach is based on the
following construction: Let $\B(\H)$ denote the algebra of bounded operators on $\H$, and let $\E$ be a closed unital subalgebra of
$\B(\H)$ containing a closed ideal $\J$. Let us assume that (i) $W_-$
belongs to $\E$, (ii) the image of $W_-$ through the quotient map $q:
\E \to \E /\J$ is a unitary operator incorporating $S$. We shall see
that in the simplest situation, $q(W_-)$ can be identified with $S$, but that in the general
case, $q(W_-)$ incorporates besides $S$ other components which account
for the correction in \eqref{LevMartin}.

We think of $\J$ as the algebra related to the bound states system,
and of $\E/\J$ as the one corresponding to the scattering system. By
the general machinery of $K$-theory of $C^*$-algebras the map $q$
gives rise to a topological boundary map, called the index map,
$\ind:K_1(\E/\J)\to K_0(\J)$ which can be described as follows: If
$\Gamma\in \E /\J$ is a unitary representing an element $[\Gamma]_1$ in
$K_1(\E/\J)$ and having a preimage $W \in \E$ under $q$ which is
an partial isometry, then $\ind([\Gamma]_1) = [W W^*]_0-[W^*
W]_0$, the difference of the classes in $K_0(\J)$ of the range
and the support projections of $W$. In particular if $W_-$ belongs to
$\E$, the asymptotic completeness yields
\begin{equation}\label{eq-lev-ab}
\ind\big([q(W_-)]_1\big) = - [P_{\mathrm p}]_0\ ,
\end{equation}
where $P_{\mathrm p}$ is the orthogonal projection on the subspace spanned by the eigenvectors of $H$.
This result is our abstract Levinson's theorem. Concrete Levinson's theorems like
\eqref{LevMartin} arise if we apply functionals to the $K$-groups to
obtain numbers.

For a large class of scattering systems we expect that $\J= \K(\H)$,
where $\K(\H)$ is the algebra of compact operators on $\H$, and that
$\E/\J$ is isomorphic to $C\big(\S;\K(\HH)\big)^\sim$, the continuous
functions on the circle with values in $\K(\HH)$ and a unit added.
Let us already say that this assumption holds for the model considered in this paper.
In that case $K_0(\J)$ and
$K_1(\E/\J)$ are both isomorphic to $\Z$ and so the index map in
\eqref{eq-lev-ab} reduces to an homomorphism $\Z\to\Z$, and hence to a
multiple of the identity $n\,\mbox{\rm id}$ for some $n\in\Z$.
Indeed, the trace $\Tr$ on $\H$ induces a functional
$\Tr_*:K_0\big(\K(\H)\big)\to \Z$, with $\Tr_*([P]_0)=\Tr(P)$ if $P\in \K(\H)$ is a projection, and this functional is
an isomorphism. Similarly  the winding number
$w(\Gamma)$ of the determinant $t\mapsto \det\big(\Gamma(t)\big)$ induces
a functional
$w_*:K_1\big(C\big(\S;\K(\HH)\big)^\sim\big)\to\Z$, with $w_*([\Gamma]_1)=w(\Gamma)$ for any unitary $\Gamma \in C\big(\S;\K(\HH)\big)^\sim$,
which also yields to an isomorphism.
Then there exists a $n\in\Z$ such that the following index
theorem holds
\begin{equation}\label{eq-lev2}
n w\big(q(W_-)\big) = -\Tr(P_{\mathrm p})\ .
\end{equation}
The number $n$ depends on $\J\subset \E$, that is on $\E$ and its ideal $\J$ or,
as is the technical term, on the extension defined by $\E$ and $\J$.
We will find below that for the algebras constructed in Section \ref{secal} this number $n$ is equal to $1$, meaning that we are considering the Toeplitz extension and the index theorem of Krein-Gohberg.
This is our formulation of the concrete Levinson's theorem
\eqref{LevMartin}. Note that there is room for further, potentially
unknown, identities of Levinson type by choosing other functionals
in cases in which the $K$-groups are richer than those considered
above, see for example \cite{KPR}.

There is a certain issue about the functional defined by the winding
number, as the determinant of the unitary $\Gamma(t)$ is not always defined.
Nevertheless it is possible to define $w_*$ on a $K_1$-class $[\Gamma]_1$
simply by evaluating it on a representative on which the determinant $\det(\Gamma(t))$ is well defined and depends continuously on $t$. For our purposes this is
not sufficient, however, as it is not a priori clear how to construct
for a given $\Gamma$ such a representative. We will therefore have to make
recourse to a regularization of the determinant.
Let us explain this regularization in the case that
$\Gamma(t)-1$ lies in the $p$-th Schatten ideal for some integer $p$, that is, $|\Gamma(t)-1|^p$ is
traceclass.
We denote by $\{e^{i\theta_j(t)}\}_j$ the set of eigenvalues of $\Gamma(t)$.
Then the regularized Fredholm determinant $\det_p$, defined by \cite[Chap.~XI]{GGK}
$${\det}_p\big(\Gamma(t)\big)
= \prod_j e^{i\theta_j(t)}\exp\left(\sum_{k=1}^{p-1}\frac{(-1)^k}{k}
( e^{i\theta_j(t)} -1)^k\right)
$$
is finite and non-zero.
Let us furthermore suppose that $t\mapsto \Gamma(t)-1$ is continuous in the
$p$-th Schatten norm. Then the map $t\to {\det}_p\big(\Gamma(t)\big)$ is continuous and hence
the winding number of
$\S\ni t\mapsto {\det}_p\big(\Gamma(t)\big)\in\C^*$ exists.
In addition, if we suppose that $t\mapsto \Gamma(t)$ is
continuously differentiable in norm, then one can show that in contrast to the value of the determinant at fixed $t$, the winding number
will not depend on $p$.
Indeed, with this additional assumption the map $t\mapsto {\det}_{p+1}\big(\Gamma(t)\big)$ is continuously differentiable and
\begin{equation}\label{eq-wind}
\frac{\d \ln\det_{p+1}\big(\Gamma(t)\big)}{\d t} = \tr\big[\big(1-\Gamma(t)\big)^{p}\Gamma^*(t)
\Gamma'(t)\big].
\end{equation}
Furthermore, by integrating the map $t \mapsto \tr \big[\big(1-\Gamma(t)\big)^{q}\Gamma^*(t) \Gamma'(t)\big]$ around $\S$, one can show that the result is independent of $q$ as long as $q\geq p$. Note that since we have not been able to locate these results in the literature, we prove them in the Appendix.
Thus, in order to defined $w$ we may choose any $p$ for which $\Gamma(t)-1$ takes its values in
the $p$-th Schatten ideal and the corresponding map $t \mapsto \Gamma(t)-1$ is continuous in the $p$-th Schatten norm. If in addition this map is
continuously differentiable in norm, then the winding number is independent of the choice of $p$.
We will see below that the counter term $c$ in \eqref{LevMartin} can also be understood as resulting from such a regularization.

Let us finally described the content of this papier. In Section \ref{sec2}, we
recall some known facts on spherical harmonics, the Fourier transform
and the dilation operator. The scattering system is introduced in
Section \ref{sec3} in which we derive formula \eqref{newformula} from
the stationary representation of the wave operators. The term $\bv(A)$
is also explicitly computed. Section \ref{secal} is
devoted to the construction of the $C^*$-algebra pertaining for
potential scattering in $\R^3$. We give two alternative descriptions
of these algebras, each one having its own interest. In Section
\ref{secandsowhat}, we show how the formula \eqref{newformula} solves
the affiliation property and we derive some consequences of it. The
concordance of our approach with already known results on
$3$-dimensional potential scattering is put into evidence. Section
\ref{allwhatIdontknow} is dedicated to the derivation of explicit
computable formulas for the topological Levinson's theorem. The usual regularization and correction are
clearly explained. Finally, in Section \ref{secpoint} we recall the
result obtained in \cite{KR1} for the model of one point interaction
in $\R^3$ and illustrate our formalism with this example.

\section*{Acknowledgements}
S. Richard is supported by the Japan Society for the Promotion of Science.

\section{Spherical harmonics, Fourier transform and dilation operator}\label{sec2}

In this section we briefly recall the necessary background on spherical harmonics, the Fourier transform and the dilation operator.

Let $\H$ denote the Hilbert space $L^2(\R^3,\d x)$, and let $\H_r:=L^2(\R_+, r^2\;\!\d r)$ with $\R_+ = (0,\infty)$. The Hilbert space $\H$ can be decomposed with respect to spherical coordinates $(r,\omega) \in \R_+\times \S^2$: For any $\l \in \N = \{0,1,2,\ldots\}$ and $m\in \Z$ satisfying $-\l\leq m \leq \l$, let $Y_{\l m}$ denote the usual spherical harmonics. Then, by taking into account the completeness of the family $\{Y_{\l m}\}_{\l \in \N, |m|\leq \l}$ in $L^2(\S^2,\d \omega)$, one has the
canonical decomposition
\begin{equation}\label{decomposition}
\H = \bigoplus_{\l \in \N, |m|\leq \l} \H_{\l m}\ ,
\end{equation}
where $\H_{\l m}=\{f \in \H \mid f(r \omega)=g(r) Y_{\l m}(\omega) \hbox{ a.e.~for some }g \in \H_r\}$. For fixed $\l \in \N$ we  denote by $\H_\l$ the subspace of $\H$ given by $\bigoplus_{-\l \leq m \leq \l}\H_{\l m}$.

Now, let $C_c^\infty(\R^3)$ denote the set of smooth functions on $\R^3$ with compact support, and let $\FF$ be the usual Fourier transform explicitly given on any $f\in C_c^\infty(\R^3)$ and for $k \in \R^3$ by
\begin{equation*}
[\FF f](k)\equiv \hat{f}(k)= (2\pi)^{-3/2} \int_{\R^3}f(x)\;\!e^{-ix\cdot k}\;\!\d x\ .
\end{equation*}
It is known that both $\FF$ and and its adjoint $\FF^*$ leave the subspace $\H_{\l m}$ of $\H$ invariant. More precisely, for any $g \in C_c^\infty(\R_+)$ and for $(\kappa,\omega) \in \R_+\times \S^2$ one has:
\begin{equation}\label{surFl}
[\FF(gY_{\l m})](\kappa \omega) = (-i)^\l \;\!Y_{\l m}(\omega) \int_{\R_+}r^2\;\! \frac{J_{\l+1/2}(\kappa r)}{\sqrt{\kappa r}}g(r)\;\!\d r\ ,
\end{equation}
where $J_\nu$ denotes the Bessel function of the first kind and of order $\nu$. So, we naturally set $\FF_\l: C_c^\infty(\R_+) \to \H_r$ by the relation $\FF(g Y_{\l m}) = \FF_\l(g)Y_{\l m}$ (it is clear from \eqref{surFl} that this operator does not depend on $m$ ). Similarly to the Fourier transform on $\R^3$, this operator  extends to a unitary operator from $\H_r$ to $\H_r$.

Let us now consider the unitary dilation group $\{U_\tau\}_{\tau \in \R}$ defined on any $f \in \H$ and for $\tau \in \R$ by
\begin{equation*}
[U_\tau f](x) = e^{3\tau/2} f(e^\tau x)\ .
\end{equation*}
Its self-adjoint generator $A$ is formally given by $\frac{1}{2}(Q\cdot P + P\cdot Q)$, where $Q=(Q_1,Q_2,Q_3)$ stands for the position operator and $P=(P_1,P_2,P_3)\equiv(-i\partial_1,-i\partial_2,-i\partial_3)$ denotes its conjugate operator. All these operators are essentially self-adjoint on the Schwartz space $\SS(\R^3)$ on $\R^3$.

It is easily observed that the formal equality $\FF\;\!A\;\!\FF^*=-A$ holds. More precisely, for any essentially bounded function $\varphi$ on $\R$, one has $\FF\varphi(A)\FF^* = \varphi(-A)$. Furthermore, since $A$ acts only on the radial coordinate, the operator $\varphi(A)$ leaves each $\H_{\l m}$ invariant. For that reason, we can consider a slightly more complicated operator than $\varphi(A)$. For any $\l \in \N $, let $\varphi_{\l}$ be an essentially bounded function on $\R$. Assume furthermore that the family $\{\varphi_{\l}\}$ is bounded. Then the operator $\bv(A):\H \to \H$ defined on $\H_{\l}$ by $\varphi_\l (A)$ is a bounded operator. Note that we use the same notation for the generator of dilations in $\H$ and for its various restrictions in suitable invariant subspaces.

Let us finally state a result about the Mellin transform.

\begin{lemma}\label{Jensen}
Let $\varphi$ be an essentially bounded function on $\R$ such that its inverse Fourier transform is a distribution on $\R$. Then, for any $f \in C^\infty_c(\R^3\setminus \{0\})$ one has
\begin{equation*}
[\varphi(A)f](r\omega) =
\frac{1}{\sqrt{2\pi}} \int_0^\infty\check{\varphi}
\big(\ln\big(r/s)\big)\;\!(r/s)^{1/2}\;\!f(s\omega)\;\! \frac{\d s}{r}\ ,
\end{equation*}
where the r.h.s.~has to be understood in the sense of distributions.
\end{lemma}

\begin{proof}
The proof is a simple application for $n=3$ of the general formulas developed in \cite[p.~439]{Jen}.
Let us however mention that the convention of this reference on the minus sign for the operator $A$ in its spectral representation
has not been followed.
\end{proof}

In particular, if $f\in \H_\l$ and $f(r\omega)=g(r)Y_{\l m}(W)$ for some $g \in C^\infty_c(\R_+)$,
then $\varphi(A)f = [\varphi(A)g]Y_{\l m}$ with
\begin{equation}\label{formuleJen}
[\varphi(A)g](r) =
\frac{1}{\sqrt{2\pi}} \int_0^\infty\check{\varphi}
\big(\ln\big(r/s)\big)\;\!(r/s)^{1/2}\;\!g(s)\;\! \frac{\d s}{r}\ ,
\end{equation}
where the r.h.s.~has again to be understood in the sense of distributions.

\section{New formulas for the wave operators}\label{sec3}

In this section we introduce the precise framework of our investigations and deduce a new formula for the wave operators. All preliminary results can be found in the paper \cite{I} or in the books \cite{AJS,RS}.

Let us denote by $H_0$ the usual Laplace operator $-\Delta$ on $\R^3$, with domain $\H^2$, the Sobolev space of order $2$ on $\R^3$. For the perturbation, we assume that $V$ is a real function on $\R^3$ satisfying for all $x \in \R^3$
\begin{equation}\label{condV}
|V(x)| \leq c\;\! \langle x \rangle^{-\beta}
\end{equation}
for some $c>0$ and $\beta >3$, where $\langle x \rangle = (1+x^2)^{1/2}$. In that situation, the operator $H=H_0+V$ is self-adjoint with domain $\H^2$. The spectrum of $H$  consists of an absolutely continuous part equal to $[0,\infty)$ and of a finite number of eigenvalues which are all located in $(-\infty,0]$.
Furthermore, the wave operators
\begin{equation*}
W_\pm := s-\lim_{t \to \pm \infty} e^{iHt}\;\! e^{-iH_0 t}
\end{equation*}
exist and are asymptotically complete. Thus, the scattering operator $S=W_+^*\;\!W_-$ is unitary and the isometries  $W_\pm$ have support and range projections
\begin{equation}\label{eq-om}
W_\pm^*W_\pm = 1,\qquad W_\pm W_\pm^*= 1-{P_{\mathrm p}} \ ,
\end{equation}
where $P_{\mathrm p}$ is the projection on the subspace of $\H$ spanned by the eigenvectors of $H$.

Now, let $\U : \H \to \int^{\oplus}_{\R_+}\HH \d \lambda$, with $\HH:=L^2(\S^2)$, be the unitary transformation that diagonalizes $H_0$, {\it i.e.}~that satisfies $[\U H_0 f](\lambda, \omega) = \lambda [\U f](\lambda,\omega)$, with $f$ in the domain of $H_0$, $\lambda \in \R_+$ and $\omega \in \S^2$. Since the operator $S$ commutes with $H_0$, there exists a family $\{S(\lambda)\}_{\lambda \in \R_+}$ of unitary operators in $\HH$ satisfying  $\U S\U^* = \{ S(\lambda)\}_{\lambda \in \R_+}$ for almost every $\lambda \in \R_+$. Furthermore, the operator $S(\lambda)-1$ is Hilbert-Schmidt in $\HH$, and the family is continuous in $\lambda\in \R_+$ in the Hilbert-Schmidt norm.

We shall now motivate the new formula for the wave operators and state the precise conjecture. Since our analysis is based on their representations in terms of the generalized eigenfunctions $\Psi_\pm$, we first recall their constructions. A full derivation can be found in \cite[Chap.~10]{AJS} or in \cite{I}. For simplicity, we shall restrict ourselves to the study of $W_-$ (and thus $\Psi_-$) and deduce the form of $W_+$ as a corollary. Because of our hypotheses on $V$, the Lipmann-Schwinger equation
\begin{equation}\label{LS}
\Psi_-(x,k)=  e^{ik\cdot x}-\frac{1}{4\pi}\int_{\R^3} \frac{e^{i|k|\;\!|x-y|}}{|x-y|}\;\!V(y)\;\! \Psi_-(y,k)\;\!\d y
\end{equation}
has for all $k \in \R^3$ with $k^2 \neq 0$ a unique solution with the second term in the r.h.s. of \eqref{LS} in $C_0(\R^3)$ for the $x$-variable. Furthermore, the following asymptotic development holds:
\begin{equation}\label{asympt}
\Psi_-(x,k) =  e^{ik\cdot x} - i (2\pi)^{3/2}\frac{1}{\sqrt{2\pi}}\;\!
\frac{e^{ i |k|\;\!|x|}}{|k|\;\!|x|}\;\!\big(S (k^2,\omega_x,\omega_k)-1\big) + \rho(x,k),
\end{equation}
where $\omega_x = \frac{x}{|x|}$, $\omega_k = \frac{k}{|k|}$ and $S(k^2,\omega_x,\omega_k)-1$ is the kernel of the Hilbert Schmidt operator $S(k^2)-1$.
For fixed $k$, the term $\rho(x,k)$ is $o(|x|^{-1})$ as $|x|\to\infty$.

Now, it is known that the wave operator can also be expressed in terms of the generalized eigenfunctions. On any $f \in \H$ the relation reads
\begin{equation}\label{waveop}
[W_- f](x) =(2\pi)^{-3/2} \int_{\R^3}\Psi_-(x,k)\;\!\hat{f}(k)\;\!\d k
\end{equation}
where the integral has to be understood as a limit in the mean \cite[Thm.~7]{I}.
By comparing \eqref{asympt} with \eqref{waveop} one is naturally led to define for all $f \in \SS(\R^3)$ and $x \neq 0$ the integral operator
\begin{equation}\label{defdeT}
[Tf](x)=-i \frac{1}{\sqrt{2\pi}}\int_{\R_+}
\frac{e^{ i \kappa\;\!|x|}}{\kappa\;\!|x|}\;\![\FF(f)](\kappa \;\!\omega_x)\;\!\kappa^2 \d \kappa\ .
\end{equation}
An easy computation shows that this operator is invariant under the action of the dilation group.
Moreover, this operator is reduced by the decomposition \eqref{decomposition}.
The following proposition contains the precise statement
of these heuristic considerations. For that purpose, let us define for each $\l \in \N$ the operator $T_\l$ acting on any $g \in C_c^\infty(\R_+)$ as
\begin{equation*}
[T_\l \;\!g](r)=-i \frac{1}{\sqrt{2\pi}}\int_{\R_+} \kappa^2 \;\!
\frac{e^{ i\kappa\;\!r}}{\kappa\;\!r}\;\! [\FF_\l \;\!g ](\kappa)\;\!\d \kappa\ .
\end{equation*}

\begin{proposition}\label{tralala}
The operator $T_\l$ extends continuously to the bounded operator $\varphi_\l(A)$ in $\H_r$ with $\varphi_\l\in C\big([-\infty,\infty],\C\big)$ given explicitly for every $x\in \R$ by
$$
\varphi_\l(x)=
\frac{1}{2}e^{-i\pi \l /2}
\frac{\Gamma\big(\frac{1}{2}(\l+3/2+ix)\big)}{\Gamma\big(\frac{1}{2}(\l+3/2-ix)\big)}  \frac{\Gamma\big(\frac{1}{2}(3/2-ix)\big)}{\Gamma\big(\frac{1}{2}(3/2+ix)\big)}  \big(1 + \tanh(\pi x)-i\cosh(\pi x)^{-1}\big)\ .
$$
and satisfying $\varphi_\l(-\infty)=0$ and $\varphi_\l(\infty)=1$.
Furthermore, the operator $T$ defined in \eqref{defdeT} extends continuously to the operator $\bv (A)\in \B(\H)$ acting as $\varphi_\l(A)$ on $\H_\l$.
\end{proposition}

\begin{proof}
For any $g \in C^\infty_c(\R_+)$ one has
\begin{eqnarray}\label{horrible}
\nonumber [T_\l\;\! g](r) & = & -i \frac{1}{\sqrt{2\pi}} \int_{\R_+} \kappa^2\;\!\frac{e^{i \kappa\;\!r}}{\kappa\;\!r}\;\!
\Big[(-i)^\l\;\!\int_{\R_+}s^2 \frac{J_{\l+1/2}(\kappa s)}{\sqrt{\kappa s}}g(s)\;\!\d s\Big]\;\!\d \kappa \\
\nonumber & = &\frac{1}{\sqrt{2\pi}}\int_{\R_+}
\Big[-i(-i)^\l\;\!\frac{r}{s}
\int_{\R_+}\kappa \;\!\frac{e^{i\kappa r/s}}{\sqrt{\kappa r/s}}\;\!J_{\l+1/2}(\kappa)\;\! \d \kappa
\Big]\big(\frac{s}{r}\big)^{1/2}\;\!g(s)\;\! \frac{\d s}{r} \\
& = & \frac{1}{\sqrt{2\pi}}\int_{\R_+}
\Big[(-i)^\l\;\!\sqrt{\frac{\pi}{2}}\;\!\frac{r}{s}
\int_{\R_+}\kappa \;\!H^{(1)}_{1/2}\big(\frac{r}{s}\kappa\big)\;\!J_{\l+1/2}(\kappa)\;\! \d \kappa
\Big]\big(\frac{s}{r}\big)^{1/2}\;\!g(s)\;\! \frac{\d s}{r}\ ,
\end{eqnarray}
where $H^{(1)}_{1/2}$ is the Hankel function of the first kind and of order $1/2$. Let us stress that the second and the third equalities have to be understood in the sense of distributions on $\R_+$, {\it cf.}~\cite{KR4}.

Now, by comparing \eqref{horrible} with \eqref{formuleJen} one observes that the operator $T_\l$ is equal on  dense set in $\H_r$
to an operator $\varphi_\l(A)$ for a function $\varphi_\l$ whose inverse Fourier transform is the distribution which satisfies for $y \in \R$:
\begin{equation*}
\check{\varphi}_\l(y)=\frac{1}{2}\;\!\sqrt{2\pi}\;\!e^{-i\pi\l/2}\;\!\!e^y
\int_{\R_+}\kappa \;\!H^{(1)}_{1/2}\big(e^y\kappa\big)\;\!J_{\l+1/2}(\kappa)\;\! \d \kappa\ .
\end{equation*}
The Fourier transform of this distribution can be computed. Explicitly, one has in the sense of distributions :
\begin{eqnarray*}
\varphi_\l(x) &=&\frac{1}{2}\;\!e^{-i\pi\l/2} \int_\R \Big[
e^{-ixy}\;\!e^y \int_{\R_+}\kappa  \;\!H^{(1)}_{1/2}\big(e^y\kappa\big) \;\!J_{\l+1/2}(\kappa)\;\! \d \kappa \Big] \d y\\
&=& \frac{1}{2}\;\!e^{-i\pi\l/2} \int_{\R_+} \Big[\kappa \;\!J_{\l+1/2}(\kappa)\int_{\R_+}\frac{\d s}{s} s^{-ix+1}\;\! \kappa^{ix-1}\;\!H^{(1)}_{1/2}(s)\Big]\d \kappa \\
&=& \frac{1}{2}\;\!e^{-i\pi\l/2} \int_{\R_+} \kappa^{(1+ix)-1}\;\!J_{\l+1/2}(\kappa)\;\!\d \kappa
\int_{\R_+} s^{(1-ix)-1}\;\!H^{(1)}_{1/2}(s)\;\!\d s \\
&=& \frac{1}{2\pi}\;\!e^{-i\pi(\l+1/2)/2} \;\!e^{\pi x/2}\;\!
\frac{\Gamma\big(\frac{1}{2}(\l+3/2+ix)\big)}{\Gamma\big(\frac{1}{2}(\l+3/2-ix)\big)}  \;\!\Gamma\big((3/2-ix)/2\big)\;\!\Gamma\big((1/2-ix)/2\big)\ .
\end{eqnarray*}
The last equality is obtained by taking into account the relation between the Hankel function $H^{(1)}_\nu$ and the Bessel function $K_\nu$ of the second kind as well as the Mellin transform of the functions $J_\nu$ and the function $K_\nu$ as presented in \cite[Eq.~10.1 \& 11.1]{O}.
Then, by taking into account the equality
$\Gamma(z)\;\!\Gamma(1-z) = \frac{\pi}{\sin(\pi z)}$
for $z=\frac{1}{4}-i\frac{x}{2}$ one obtains that
\begin{eqnarray*}
&&\frac{1}{2\pi}\;\!e^{-i\pi/4}\;\!e^{\pi x/2}
\;\!\Gamma\big((3/2-ix)/2\big)\;\!\Gamma\big((1/2-ix)/2\big)\\
&=& \frac{\Gamma\big(\frac{1}{2}(3/2-ix)\big)}{\Gamma\big(\frac{1}{2}(3/2+ix)\big)}  \frac{e^{i\pi/4}\;\!e^{\pi x/2}}{e^{i\pi/4}e^{\pi x/2} - e^{-i\pi /4} e^{-\pi x/2}} \\
&=& \frac{1}{2} \frac{\Gamma\big(\frac{1}{2}(3/2-ix)\big)}{\Gamma\big(\frac{1}{2}(3/2+ix)\big)}  \big(1 + \tanh(\pi x)-i\cosh(\pi x)^{-1}\big)
\end{eqnarray*}
which leads directly to the desired result.
The asymptotic values of $\varphi_\l$ can easily be obtained by using the asymptotic development of the function $\Gamma$ as presented in \cite[Eq.~6.1.39]{AS}.

For the final statement, it is sufficient to observe that the operator $\varphi_\l(A)$ is clearly well defined in $\H_{\l m}$, but since it does not depend on $m$ this operator is also well defined in $\H_\l$.
Furthermore, since the norm of the operator $\varphi_\l(A)$ is equal to $1$ the collection $\{\varphi_\l(A)\}$ defines a bounded operator in $\H$ equal to the operator $T$ on a dense set.
\end{proof}

\begin{corollary}\label{lemsur0}
One has $\varphi_0(x)= \frac{1}{2}\big(1 + \tanh(\pi x)-i\cosh(\pi x)^{-1}\big)$.
\end{corollary}

By collecting these results, one can finally state our main conjecture and its consequence :

\begin{assumption}\label{hyp}
The kernel $\rho$ in \eqref{asympt} defines a compact operator $K$ on $\H$ by the formula $[Kf](x)=(2\pi)^{-3/2}\int_{\R^3}\rho(x,k)\;\!\hat{f}(k)\;\!\d k$ for any  $f \in  \H$.
\end{assumption}

\begin{theorem}\label{main}
Let $V$ satisfy condition \eqref{condV} and assume that Conjecture \ref{hyp} is verified. Then the wave operator $W_-$ for the scattering system $(-\Delta + V, -\Delta)$ satisfies the equality
\begin{equation}\label{W-}
W_- =1+ \bv(A)(S -1) + K
\end{equation}
with $\bv(A)$ defined in Proposition \ref{tralala}.
\end{theorem}

\begin{proof}
Formulas \eqref{asympt} and \eqref{waveop} together with Proposition \ref{tralala} and Conjecture \ref{hyp} lead directly to the result.
\end{proof}

\begin{corollary}
If the hypotheses of Theorem \ref{main} are satisfied, then
\begin{equation*}
W_+ =1+ \big(1-\bv(A)\big)(S^*-1) + K
\end{equation*}
with $K$ a compact operator in $\H$.
\end{corollary}

\begin{proof}
Starting from the relation $W_-^*W_+ = S^*$, one deduces from the asymptotic completeness that
$$
W_+ = W_-W_-^*W_+ = W_-S^*\ .
$$
By taking the expression \eqref{W-} into account, one readily obtains the result.
\end{proof}

\section{Constructing the algebras}\label{secal}

In this section we define the $C^*$-algebras suitable for the
scattering system introduced in the previous section. In a different context, these algebras were studied in \cite{Georgescu} from which we recalled some preliminary results. We also  mention that very
similar algebras have already been introduced a long time ago
\cite{BC1, BC2,CH}.

The forthcoming algebras are constructed with the help of the operator $H_0$
and with the generator $A$ of dilations. The crucial property is that
$A$ and $B:=\frac{1}{2}\ln(H_0)$ satisfy the canonical
commutation relation $[A,B]=i$ so that $A$ generates translations in
$B$ and vice versa,
\begin{equation*}
e^{itB} A e^{-itB} = A+t, \quad e^{isA} B e^{-isA} = B-s.
\end{equation*}
Furthermore, recall that these operators are reduced by the decomposition \eqref{decomposition} and that the spectrum of $H_0$ is $\R_+$ and that of $A$ is $\R$. In the following paragraphs we shall freely use the isomorphism between $L^2(\R^3)$ and $L^2(\R_+;\HH)$ with $\HH = L^2(\S^2)$.

Let $\E_o$ be the closure in $\B(\H)$ of the algebra generated by
elements of the form $\eta(A)\psi(H_0)$, where $\eta$ is a norm-continuous function on $\R$ with values in $\K(\HH)$ and which has limits at $\pm \infty$, and
$\psi$ is a norm-continuous function $\R_+$ with values in $\K(\HH)$ and which has limits at $0$ and
at $+\infty$. Stated differently, $\eta\in C\big(\bR;\K(\HH)\big)$, where $\bR=[-\infty,+\infty]$, and $\psi \in C\big(\bRp;\K(\HH)\big)$ with
$\bRp=[0,+\infty]$, the continuity referring to the norm topology on $\K(\HH)$. Let $\J$ be the norm closed
algebra generated by $\eta(A)\psi(H_0)$ with similar functions $\eta$ and
$\psi$ for which the above limits vanish. Clearly, $\J$ is a closed ideal in $\E_o$ and is equal to the algebra $\K(\H)$ of compact operators in $\H$. These statements and the following ones follow from \cite[Sec.~3.5]{Georgescu} via the Mellin transform.

To describe the quotient $\E_o/\J$ we consider the square $\blacksquare:=\bRp\times \bR$ whose boundary $\square$ is the
union of four parts: $\square =B_1\cup B_2\cup B_3\cup
B_4$, with $B_1 = \{0\}\times \bR$, $B_2 = \bRp \times \{+\infty\}$, $B_3 = \{+\infty\}\times \bR$ and $B_4 = \bRp\times \{-\infty\}$. We can also view $C(\square)$ as the subalgebra of
\begin{equation*}
C\big(|\underline{\overline{\hbox{\phantom{\c L}}}}| \big) :=
C(\bR)\oplus C(\bRp)\oplus C(\bR)\oplus C(\bRp)
\end{equation*}
given by elements
$\Gamma:=(\Gamma_1,\Gamma_2,\Gamma_3,\Gamma_4)$ which coincide at the
corresponding end points, that is, for instance,
$\Gamma_1(+\infty)=\Gamma_2(0)$. Then $\E_o/\J$ is isomorphic to
$C\big(\square;\K(\HH)\big)$, and if we denote the quotient map by
$$
q: \E_o\to \E_o/\J \cong  C\big(\square;\K(\HH) \big)
$$ then the image $q\big(\eta(A)\psi(H_0)\big)$ in $C\big(|\underline{\overline{\hbox{\phantom{\c L}}}}| ;\K(\HH)\big)$
is given by $\Gamma_1 = \eta(\cdot)\psi(0)$, $\Gamma_{2} =
\eta(+\infty)\psi(\cdot)$,
$\Gamma_{3} = \eta(\cdot)\psi(+\infty)$ and $\Gamma_{4} =
\eta(-\infty)\psi(\cdot)$.

By construction the algebra $\E_o$ does not contain a unit. For a non-unital algebra $\A$, we generically write $\A^\sim$ for the algebra $\A$ with a unit added.  So let us set $\E:=\E_o^\sim$ for the closed $C^*$-algebra generated by $\E_o$ and by the operator $1\in \B(\H)$. Then the quotient algebra $\Q:=\E/\J$ is also unital and can be identified with the algebra $C\big(\square ;\K(\HH)\big)^\sim$, the algebra generated by $C\big(\square ;\K(\HH)\big)$ and by the constant function $1$ on $\square$.
Clearly, this algebra is isomorphic to the algebra $C\big(\S;\K(\HH)\big)^\sim$.
Thus, the algebras introduced so far correspond to the one presented in the introduction and it follows that the trace $\Tr$ on $\H$ induces a functional
$\Tr_*:K_0(\K(\H))\to \Z$,
 $\Tr_*([P]_0) = \Tr(P)$ for any projection $P \in \K(\H)$ and this functional is
an isomorphism. Similarly the winding number $w(\Gamma)$ of the determinant $\square\ni t\mapsto \det\big(\Gamma(t)\big)$ induces
a functional
$w_*:K_1(\Q)\to\Z$, $w_*([\Gamma]_1)= w(\Gamma)$ for any unitary $\Gamma \in \Q$
which also yields to an isomorphism.
As a consequence, the groups $K_0(\J)$ and $K_1(\Q)$ are both isomorphic to $\Z$
and so the index map in \eqref{eq-lev-ab} reduces to an homomorphism $\Z\to\Z$, and hence to a multiple of the identity $n\,\mbox{\rm id}$ for some $n\in\Z$.
In the next proposition we show that the factor $n$ is equal to $1$ for the algebras introduced above.

\begin{proposition}\label{thm-norm}
For the extension defined by $\E$ and $\J$
the factor $n$ in the index map in \eqref{eq-lev2} is equal to $1$.
\end{proposition}

\begin{proof} Consider the index map $\ind:K_1(\Q)
{\to} K_0(\J)$ relevant to our algebras, with $\J = \K(\H)$ and $\Q$ isomorphic to the unital algebra $C\big(\S;\K(\HH)\big)^\sim$.
Upon identifying $K_0\big(\K(\H)\big)\cong\Z$ via $\Tr_*$ and
$K_1\big(\big(C(\S;\K(\HH)\big)^\sim\big)\cong\Z$
via $w_*$ we get a group homomorphism
$\Tr_*\circ\mbox{\rm ind}\circ w_*^{-1}: \Z \to \Z$ and hence
$\Tr_*\circ\mbox{\rm ind}\circ w_*^{-1} = n\,\mbox{\rm id}$ for some $n$.
Hence $n$ is determined by the equation
$\Tr_*\big(\ind([\Gamma]_1)\big) = nw_*([\Gamma]_1)$ which
must hold for all elements of $K_1(\Q)$.
We know furthermore that $\Tr_*\big(\ind([\Gamma]_1)\big) = -\mbox{\rm index}(W)$ provided $W\in \E$
is a lift of $\Gamma$ which is a partial isometry. Taking
$W=W^{\alpha}_-$, the wave operator for some point interaction $\alpha<0$ recalled in Section \ref{secpoint}, we infer from our explicit calculation that
$\mbox{\rm index}(W^{\alpha}_-)=-1$ and
$w\big(q(W^{\alpha}_-)\big)=1$. Hence $n=1$.
\end{proof}

Let us now present an alternative description of the $C^*$-algebra $\E$. As already mentioned, the algebra $\E_o$ has been introduced and thoroughly studied in another context in \cite[Sec.~3.5]{Georgescu}. All the proofs of the following statements can be mimicked from the corresponding ones in that paper. We also use the convention of that reference, that is: if a symbol like $T^{(*)}$ appears in a relation, it means that this relation holds for $T$ and for its adjoint $T^*$. The function $\chi$ denotes the characteristic function.

\begin{lemma}\label{appartenance}
An operator $W$ belongs to $\E$ if and only if there exist
$\Gamma_{1}, \Gamma_3\in C\big(\bR;\K(\HH)\big)^\sim$ and $\Gamma_{2},\Gamma_{4} \in C\big(\bRp;\K(\HH)\big)^\sim$
such that the following conditions are satisfied:
\begin{enumerate}
\item[{\rm (i)}] $\lim_{\e \to 0}\|\chi(H_0\leq \e) \;\!(W -\Gamma_1(A))^{(*)}\|=0$, and
$\lim_{\e \to +\infty}\|\chi(H_0\geq \e) \;\!(W -\Gamma_3(A))^{(*)}\|=0$,
\item[{\rm (ii)}] $\lim_{t \to -\infty}\|\chi(A \leq t) \;\!(W -\Gamma_4(H_0))^{(*)}\|=0$, and
$\lim_{t \to +\infty}\|\chi(A \geq t) \;\!(W-\Gamma_2(H_0))^{(*)}\|=0$.
\end{enumerate}
\end{lemma}

Let us note that conditions {\rm (i)} and {\rm (ii)} can also be rewritten as
\begin{equation*}
\lim_{t\to -\infty}\|\chi(H_0\leq 1) \;\!U^A_t\;\!(W -\Gamma_1(A))^{(*)}\;\!U^A_{-t}\|=0, \quad
\lim_{t \to +\infty}\|\chi(H_0\geq 1) \;\!U^A_t\;\!(W -\Gamma_3(A))^{(*)}\;\!U^A_{-t}\|=0\ ,
\end{equation*}
\begin{equation*}
\lim_{t\to -\infty}\|\chi(A \leq 0)\;\!U^B_{-t}\;\!(W -\Gamma_4(H_0))^{(*)}\;\!U^B_t\|=0, \quad
\lim_{t \to +\infty}\|\chi(A \geq 0)\;\!U^B_{-t}\;\!(W -\Gamma_2(H_0))^{(*)}\;\!U^B_t\|=0\ ,
\end{equation*}
where $U^A_t=e^{-itA}$, $U^B_t=e^{-itB}$ and $B=\frac{1}{2}\ln(H_0)$. It also follows from these conditions that one necessarily has in the strong topology  $s-\lim_{t\to \pm \infty} U^A_t\;\!W\;\!U^A_{-t} = \Gamma_{3/1}(A)$ and $s-\lim_{t\to \pm \infty}U^B_{-t}\;\!W \;\!U^B_t = \Gamma_{2/4}(H_0)$.
As a final remark in this section, let us mention that $W$ is a compact operator in $\H$ if and only if it satisfies conditions {\rm (i)} and {\rm (ii)} of Lemma \ref{appartenance} with $\Gamma_j = 0$ for $j \in \{1,2,3,4\}$. And so, such a property holds for all elements in $\J$.

\section{The affiliation property and its consequences}\label{secandsowhat}

In this section, we show how the affiliation property is solved by formula \eqref{newformula}. More precisely, we prove that the operator $W_-$ belongs to the algebra $\E$, and derive the expressions for the related operators $\Gamma_j$ with $j \in \{1,2,3,4\}$. In fact, all these operators are already known, but our approach gives a global framework, and provides a stronger convergence to them.

As already mentioned, the scattering operator $S$ is a function of the Laplace operator with a function $\R_+\ni \lambda\mapsto S(\lambda)-1 \in \K(\HH)$ that is continuous in the Hilbert-Schmidt norm. A fortiori, this map is continuous in the norm topology on $\K(\HH)$, and in fact the map $\lambda \mapsto S(\lambda)$ belongs to $C\big(\bRp;\K(\HH)\big)^\sim$. Indeed, it is well known that $S(\lambda)$ converges to $1$ as $\lambda \to \infty$, see for example \cite[Prop.~12.5]{AJS}. For the low energy behavior, see \cite{JK}, where the norm convergence of the $S(\lambda)$ for $\lambda \to 0$ is proved under a more restrictive condition on the potential: $\beta>5$ in \eqref{condV}. The picture is the following: If $H$ does not possess a $0$-energy resonance, then $S(0)$ is equal to $1$, but if such a resonance exists, then $S(0)$ is equal to $1-2P_{00}$, where  $P_{00}$ denotes the orthogonal projection on the one-dimensional subspace of $\HH$ spanned by $Y_{00}$.

Now, the operator $\bv(A)$ in \eqref{newformula} is a function of the generator $A$ of dilations with a function $\bv$ which belongs to $C\big(\bR, \B(\HH)\big)$, where $\B(\HH)$ is endowed with the strong operator topology. Indeed, this easily follows from Proposition \ref{tralala}. However, once multiplied with the operator $\big(S(0)-1\big)$, the map $\R \ni \xi \mapsto \bv(\xi)\big(S(0)-1\big) \in \K(\HH)$ is norm continuous and admits limits at $\pm \infty$. In fact, if $H$ does not possess a $0$-energy resonance then this map is trivial, but if such a resonance exists then the corresponding map is not trivial since $S(0)-1\neq 0$.

By collecting these information one can now prove the main result of this section :

\begin{proposition}\label{lemasym}
Let $V$ satisfy condition \eqref{condV} with $\beta>5$ and assume that Conjecture \ref{hyp} is verified.
Then the wave operator $W_-$ belongs to $\E$ and its image through the quotient map $q: \E\to \E/\J$ is a unitary element of $C\big(\square;\K(\HH) \big)^\sim$ given by
\begin{equation}\label{quadruple}
\big(1+\bv(\cdot)
\big(S(0)-1\big), S(\cdot),1,1\big) \ .
\end{equation}
In other words, $\Gamma_1 =1+\bv(\cdot)\big(S(0)-1\big)$, $\Gamma_2=S(\cdot)$, $\Gamma_3 =1$ and $\Gamma_4 = 1$.
\end{proposition}

\begin{proof}
It has been proved in Theorem \ref{main} that $W_- =1+ \bv(A)(S -1) + K$. Clearly, the operators $1$ and $K$ belong to $\E$ and thus one  only has to show that $\bv(A)(S-1)\in \E$.
As mentioned above, the map $\lambda \mapsto S(\lambda)-1$ belongs to $C\big(\bRp;\K(\HH)\big)$ but the map $\xi \mapsto \bv(\xi)$ only belongs to $C\big(\bR, \B(\HH)\big)$. However, this lack of compactness for the image of $\bv(\cdot)$ does not bother since
$\bv(A)$ is multiplied by $S-1$ and, for the angular part the multiplication of a bounded operator with a compact operator is compact. More precisely, let us work in the spectral representation of $A$ which is obtained via the Mellin transform. The corresponding Hilbert space is $L^2(\R;\HH)$, and in that representation the operator $A$ corresponds to the multiplication by the variable $Q$ and the operator $B=\frac{1}{2}\ln(H_0)$ corresponds to its conjugate operator $P$. Clearly, this Hilbert space is isomorphic to $L^2(\R)\otimes \HH$ and let $\U$ denote the final isomorphism between $\H$ and the latter space. We recalled this construction in order to use a result from \cite[Sec.~3.5]{Georgescu}, namely
that $\U\E\U^{-1}$ is equal to $C^*[\eta(Q)\psi(P)]\otimes \K(\HH)$ with $C^*[\eta(Q)\psi(P)]$ the closure in $\B\big(L^2(\R)\big)$ of the algebra generated by products of the form $\eta(Q)\psi(P)$ with $\eta,\psi \in C(\bR)$.
Then, one concludes by observing that $\U(S-1)\U^{-1}= \tilde{\psi}(P)$ with $\tilde{\psi}\in C(\bR)\otimes\K(\HH)$ and $\U\bv \U^{-1}= \tilde{\eta}(Q)$ with $\tilde{\eta} \in C(\bR)\otimes\B(\HH)$, which means that the product $\U\bv(A)(S-1)\U^* $ belongs to $C^*[\eta(Q)\psi(P)]\otimes \K(\HH)$. One has thus obtained $\bv(A)(S-1) \in \E$.

Now, the unitarity of $q(W_-)$ follows from the fact that $W_-$ is a partial isometry with finite kernel and co-kernel, and thus its image through the quotient by the compact operators is a unitary element of the Calkin algebra. A fortiori this image is also unitary in the smaller algebra $\Q$. Finally, the asymptotic operators are easily calculated. Since the compactor operator $K$ does not give any contribution for them, one has to take care of the contributions coming from the restrictions of $1+\bv(A)(S-1)$ as explained in Section \ref{secal}.
\end{proof}

Let us now explain why the asymptotic operators $\Gamma_1(A),\Gamma_3(A)$ and $\Gamma_2(H_0),\Gamma_4(H_0)$ are very natural, and how they could be guessed. It is mentioned at the end of Section \ref{secal} that if $W_-$ belongs to $\E$, then one necessarily has $\Gamma_{2/4}(H_0) = s-\lim_{t\to \pm \infty}U^B_{-t}\;\!W_- \;\!U^B_t$. But the intertwining relation for the wave operator, the invariance principle and the asymptotic completeness imply that
$s-\lim_{t\to -\infty}U^B_{-t}\;\!W_- \;\!U^B_t = W_-^* \;\! W_-$ and that $s-\lim_{t\to +\infty}U^B_{-t}\;\!W_- \;\!U^B_t = W_+^* \;\! W_-$. The former is equal to $1$ by the relation \eqref{eq-om}, and the latter is by definition the scattering operator. What we want to emphasize in
this paper is that the convergences to these operators do not only hold in the strong topology, but in the stronger topology indicated by Lemma \ref{appartenance}.

The operator $\Gamma_3(A)$ corresponds to the asymptotic of the wave operator $W_-$ at high energy. Heuristically, it is not surprising that the wave operator is close to the identity at high energy. In fact, statements like
\begin{equation*}
(W_--1)\;\!\chi(H_0\geq \varepsilon)\;\!\chi(A\geq 0) \in \K(\H) \qquad \hbox{for all } \varepsilon >0\ ,
\end{equation*}
(see for example \cite{Enss,P}) are a weaker formulation of both our statements on the convergence of $W_-$ to $\Gamma_3(A)$ and to $\Gamma_4(H_0)$.

The operator $\Gamma_1(A)$ and the convergence of the wave operator to it deserves a special attention. It is easily observed that the following equality holds:
\begin{equation*}
e^{-itA}\;\!W_-(H_0+V,H_0)\;\!e^{itA} = W_-(H(t),H_0)\ ,
\end{equation*}
where $H(t) = H_0 + e^{-2t}V(e^{-t} \cdot)$. For clarity, the dependence of $W_-$ on both self-adjoint operators used to define it is mentioned. It has been proved in \cite{AGH,AGHH} that the limit, as $t \to -\infty$, of the operator $H(t)$ converges in the resolvent sense to a zero-range perturbation of the Laplacian. More precisely, if $H \equiv H(0)$ does not possess a $0$-energy resonance, then $H(t)$ converges in the resolvent sense to $-\Delta$, but if $H$ possesses a $0$-energy resonance, it converges to $H^0$, the one point perturbation of the Laplacian with the parameter equal to $0$. For completeness, the one-point perturbation systems in $\R^3$ are briefly recalled in Section \ref{secpoint}. The topology of the convergence of the resolvent depends on the presence or the absence of a $0$-energy eigenvalue: norm topology if there is no $0$-energy eigenvalue, strong topology otherwise. However, it seems to us that these convergences could still be improved by considering the operator
\begin{equation*}
\chi(H_0\leq 1)\;\!\big((H(t)-z)^{-1}-(H(-\infty)-z)^{-1}\big)
\end{equation*}
for $z \in \C\setminus \R$.

Now, it is known that if the operator $H(t)$ converges in a suitable sense as $t$ tends to $-\infty$ to an operator $H(-\infty)$, then the corresponding operator $W_-(H(t),H_0)$ also converges to $W(H(-\infty),H_0)$, also in a suitable topology, see for example \cite{BG} or \cite[Sec.~10.4.6]{K}. Obviously, if $H(-\infty)$ is the free Laplacian, then $W_-(-\Delta, H_0)=1$, and that is what is expected if $H$ as no $0$-energy resonance. But if $H$ has such a $0$-energy resonance, then $H(t)$ converges to $H^0$, and the corresponding wave operator $W_-(H^0,H_0)$ is equal to
\begin{equation}\label{contribution0}
1 - \big(1 + \tanh(\pi A)-i\cosh(\pi A)^{-1}\big)P_{00}
\end{equation}
as proved in \cite{KR1} and recalled in Section \ref{secpoint}.

Let us finally observe that these results match perfectly with what has been obtained in Proposition \ref{lemasym}. Indeed, as already mentioned, if $H$ has a $0$-energy resonance, then $S(0) = 1-2P_{00}$, and otherwise $S(0)=1$, which corresponds respectively to $\Gamma_1(A)= 1 - \big(1 + \tanh(\pi A)-i\cosh(\pi A)^{-1}\big)P_{00}$ and to $\Gamma_1(A)=1$. And as a final comment, let us note an apparently new observation on the one point interaction systems: For the parameter equal to $0$ and $\infty$, the corresponding wave operators can be seen as the restriction at energy $0$ of the wave operators for more regular Schr\"odinger operators in $\R^3$.

\section{Explicit formulas for the topological Levinson's theorem }\label{allwhatIdontknow}

In this section, we state a precise version of the topological Levinson's theorem already sketched in the introduction. For that purpose, let us start by rewriting a precise version of \eqref{eq-lev-ab}.
Recall that the quadruple $\Gamma:=\big(1+\bv(\cdot)
\big(S(0)-1\big), S(\cdot),1,1\big)$ obtained in \eqref{quadruple}
corresponds to the image of the wave operator $W_-$ in the quotient algebra $\Q\equiv C\big(\square;\K(\HH)\big)^\sim$. Recall furthermore that this quadruple is a unitary element of the mentioned algebra and that the projection $P_{\mathrm p}$ on the subspace spanned by the eigenvectors of $H$, is an element of the algebra $\J\equiv \K(\H)$.

\begin{proposition}
Let $V$ satisfy condition \eqref{condV} with $\beta>5$ and assume that Conjecture \ref{hyp} is verified. Then one has
\begin{equation*}
\ind\big(\big[\big(1+\bv(\cdot)
\big(S(0)-1\big), S(\cdot),1,1\big)\big]_1\big) = -[P_{\mathrm p}]_0
\end{equation*}
where $\ind$ is the index map from the $K_1$-group of $\Q$ to the $K_0$-group of $\J$.
\end{proposition}

Clearly, if $H$ has no $0$-energy resonance, then $S(0)-1$ is equal to $0$, and thus the quadruple $\Gamma$ can be identified with $S(\cdot)$.
But if such a resonance exists, then the contribution of the wave
operator at energy $0$, that is the operator $\Gamma_1(A)$, is not trivial and is given by \eqref{contribution0}. This allows us to obtain a concrete computable version of our topological Levinson's theorem and show
how this operator accounts for the
correction $\nu$ in \eqref{LevMartin}.

As already mentioned for the class of perturbations
we are considering the map $\R_+\ni \lambda \mapsto S(\lambda)-1 \in\K(\HH)$ is continuous in the Hilbert-Schmidt norm. Furthermore,
it is known that this map is even continuously differentiable in the norm topology.
In particular, the on-shell time delay operator
$-i\;\!S(\lambda)^*\;\!S'(\lambda)$ is well defined for each $\lambda
\in \R_+$, see \cite{Jen,Jensen83} for details.
It then follows that :

\begin{theorem}\label{Levnous}
Let $V$ satisfy condition \eqref{condV} with $\beta>5$ and assume that Conjecture \ref{hyp} is verified.
Then for any $p\geq 2$ one has
\begin{equation}\label{new}
2\pi \,\Tr[P_{\mathrm p}] =
 \int_{-\infty}^{\infty} \tr \big[i\big(1-\Gamma_1(\xi)\big)^p
\Gamma_1(\xi)^*\,\Gamma_1'(\xi)\big] \d \xi
+ \int_0^\infty  \tr \big[i\big(1-S(\lambda)\big)^p
S(\lambda)^*\,S'(\lambda)\big] \d
\lambda .
\end{equation}
If the map $\lambda \mapsto S(\lambda)-1$ is continuously differentiable even in
the Hilbert-Schmidt norm, then the above equality holds also for $p=1$.
\end{theorem}

\begin{proof}
We apply \eqref{eq-lev2} with $n=1$ (Proposition~\ref{thm-norm}). Thus
$w(q(W_-)) = -\Tr(P_{\mathrm p})$ and our aim is to determine a computable expression for the l.h.s.. Clearly, only $\Gamma_1$ and $\Gamma_2=S(\cdot)$ will contribute in that
calculation.  $\Gamma_1$ is very regular, $\Gamma_1(\xi)-1$ having finite
rank and being smooth in $\xi$. Thus
under the given hypothesis, the function $\square \ni t
\mapsto \Gamma(t)-1$ is continuous with
values in the space of Hilbert Schmidt operators on $\HH$ and is
continuously differentiable in norm.
Hence $t\mapsto {\det}_2\big(\Gamma(t)\big)$ admits a winding number and this
winding number corresponds to $w\big(q(W_-)\big)$. Moreover, $t\mapsto {\det}_p\big(\Gamma(t)\big)$
is differentiable for all $p\geq 3$ so that the winding number
can be evaluated using equation \eqref{eq-wind} :
\begin{eqnarray*}
w\big(q(W_-)\big) &=& \frac{1}{2\pi i} \int_{\square} \d\ln{\det}_p(
\Gamma) \\
&=& \frac{1}{2\pi i} \int_{-\infty}^{\infty} \tr \big[\big(1-\Gamma_1(\xi)\big)^{p-1}
\Gamma_1(\xi)^*\,\Gamma_1'(\xi)\big] \d \xi \\
&&  + \frac{1}{2\pi i}\int_0^\infty  \tr \big[\big(1-S(\lambda)\big)^{p-1}
S(\lambda)^*\,S'(\lambda)\big] \d \lambda
\end{eqnarray*}
Under the stronger hypotheses that $\lambda \mapsto S(\lambda)-1$ is
continuously differentiable even in the Hilbert-Schmidt norm already
$t\mapsto {\det}_2(\Gamma(t))$ is differentiable so that we can use
the above argument for $p=2$.
\end{proof}

Comparing this result with the form of Levinson's theorem recalled in
\eqref{LevMartin} we see that in the absence of a resonance at $0$, in
which $\Gamma_1=1$, only the $S$-term contributes but contains already
the term proportional to $c$ to be subtracted. In the presence of a
resonance at $0$ the real part of the integral of the term $\Gamma_1$ yields
$$ \Re\Big\{ \frac{1}{2\pi i} \int_{-\infty}^{\infty} \tr
\big[\big(1-\Gamma_1(\xi)\big)^p \Gamma_1(\xi)^*\;\!\Gamma_1'(\xi)\big] \d \xi\Big\} =  \frac{1}{2}$$
which accounts for the correction usually found in
Levinson's theorem. Note that only the real part of this expression is of interest since its imaginary part will cancel with the corresponding imaginary part of term involving $S(\cdot)$.

\begin{remark}
We note that the flexibility of using larger $p$ allows to
obtain from Theorem~\ref{Levnous}
\begin{equation*}
\Tr[P_{\mathrm p}] =
-\nu +\frac{1}{2\pi}\int_0^\infty
\tr \big[\rho(S(\lambda)) iS^*(\lambda)S'(\lambda)\big] \d \lambda .
\end{equation*}
where $\rho:\S\to \R$ is any function of the form $\rho(z) =
\Re\left(\sum_{p=2}^\infty a_p (1-z)^p\right)$ with $a_p\in\C$ such
that $\sum_{p=2}^\infty a_p = 1$ and $\sum_{p=2}^\infty |a_p|2^p <\infty$.
Note that $\rho(1)=0$, $\rho'(1)=0$ and $\frac{1}{2\pi}\int_\S \rho =1$.
This shows that the regularized integral on the l.h.s.\ of
\eqref{LevMartin}
corresponds to the integral of the trace of a regularized time
delay, the latter being the time delay multiplied by a not
necessarily positive density
function $\frac{1}{2\pi}\rho$ which vanishes (to second order)
at energy values in channels for which there is no scattering.
\end{remark}

\section{Point interaction}\label{secpoint}

In this short section, we briefly recall the result obtained in \cite{KR1} for the system of one point interaction in $\R^3$. For that model all calculations are explicit, that is, the exact expression for the wave operators and the scattering operator have been determined. However, the main difference with potential scattering as presented above is that for point interaction, the operator $\Gamma_3(A)$ is in general not equal to $1$.

Let us consider the operator $-\Delta$ defined on $C_c^\infty(\R^3 \setminus \{0\})$. It has deficiency indices $(1,1)$ and
all its self-adjoint extensions $H^\alpha$ can be parameterized
by an index $\alpha$ belonging to $\R\cup\{\infty\}$. This parameter determines a certain boundary condition at $0$, and $-4\pi\alpha$
also has a physical interpretation as the inverse of the scattering length. The choice $\alpha=\infty$ corresponds to the free Laplacian $-\Delta$. The operator $H^\alpha$ has a single eigenvalue for $\alpha <0$ of value $-(4\pi\alpha)^2$ but no point spectrum for $\alpha \in [0,\infty]$. Furthermore, the action of the wave
operator $W_-^\alpha$ for the couple $(H^\alpha, -\Delta)$
on any $f \in L^2(\R^3)$ is given by \cite{AGHH}~:
\begin{equation*}
\big[(W_-^\alpha-1)f\big](x)=s-\lim_{R\to \infty}(2\pi)^{-3/2}
\int_{\kappa\leq R}\kappa^2 \;\!\d \kappa \int_{\S^2} \d \omega \;\!\frac{e^{i\kappa|x|}}{(4\pi\alpha -i\kappa)|x|}
\;\!\hat{f}(\kappa\omega)\ .
\end{equation*}
It is easily observed that $W_-^\alpha-1$ acts trivially on the orthocomplement of the range of $P_{00}$. Finally, the scattering operator $S^\alpha$ for this system is given by
\begin{equation*}
S^\alpha =  1 -P_{00} + \frac{4\pi\alpha +i\sqrt{-\Delta}}{4\pi\alpha -i\sqrt{-\Delta}}P_{00}\ .
\end{equation*}

Now, it has been proved in \cite{KR1} that the wave operator can be rewritten as
\begin{equation*}
W_-^\alpha =1+ \frac{1}{2}\big(1+\tanh(\pi A)-i\cosh(\pi A)^{-1}\big)\;\!(S^\alpha-1) \;\!
P_{00}\ .
\end{equation*}
So let us set
\begin{equation*}
 r(\xi) =  -\tanh(\pi \xi) +i\cosh(\pi \xi)^{-1}\; .
\end{equation*}
and
\begin{equation*}
s^\alpha(\lambda) =  \frac{4\pi\alpha +i\sqrt{\lambda}}{4\pi\alpha -i\sqrt{\lambda}}\
\end{equation*}
As a consequence of the expression for $W_-^\alpha$, the operators $\Gamma_1(A), \Gamma_3(A)$ and $\Gamma_2(H_0),\Gamma_4(H_0)$ are all equal to $1$ on the orthocomplement of the range of $P_{00}$. So let us give in the following table the expressions of these operators restricted to $\H_{00}$. We also compute their corresponding contribution to the winding number of $q(W_-^\alpha)P_{00}$
\begin{center}
\begin{tabular}{|c|c|c|c|c|c|c|c|c|c|}
\hline
 & $\Gamma_1$ & $\Gamma_2$ & $\Gamma_3$ & $\Gamma_4$  & $w_1$
& $w_2$ & $w_3$ & $w_4$ & $w\big(q(W^\alpha_-)P_{00}\big)$ \\ \hline\hline
$\alpha < 0$ &$ 1 $&$ s^\alpha $&$ r $&$ 1 $&$ 0 $&$ -\frac{1}{2}  $&$
-\frac{1}{2}$&$0 $&$ -1 $ \\\hline
$\alpha = 0  $&$  r $&$ -1 $&$ r $&$ 1 $&$ \frac{1}{2}  $&$0$& $
-\frac{1}{2}$&$0 $&$ 0 $ \\\hline
$\alpha > 0  $&$ 1 $&$ s^\alpha $&$ r $&$ 1 $&$ 0 $&$ \frac{1}{2}  $&$
-\frac{1}{2}$&$0 $&$ 0 $ \\
\hline
$\alpha = \infty  $&$ 1 $&$ 1 $&$ 1 $&$ 1 $&$ 0 $&$ 0  $&$0$&$0 $&$ 0 $ \\
\hline
\end{tabular}
\end{center}
and we see that the total winding number of $w\big(q(W^\alpha_-)P_{00}\big)$ is equal to minus the number of bound states of $H^\alpha$. This corresponds exactly to the topological Levinson's theorem \eqref{eq-lev2}.

\section{Appendix}

Let $\H$ be an abstract Hilbert space and $\Gamma$ be a map $\S \to \U(\H)$ such that $\Gamma(t)-1 \in  \K(\H)$ for all $t \in \S$. Let $p \in \N$ and let $S_p$ denote the Schatten ideal in $\K(\H)$.

\begin{lemma}
Assume that the map $\S \ni t\mapsto \Gamma(t)-1 \in S_p$ is continuous in norm of $S_p$ and is continuously differentiable in norm of $\K(\H)$. Then the map $\S \ni \mapsto \det_{p+1}\big(\Gamma(t)\big)\in \C$ is continuously differentiable and the following equality holds:
\begin{equation}\label{aobtenir}
\frac{\d \ln\det_{p+1}\big(\Gamma(t)\big)}{\d t} = \tr\big[\big(1-\Gamma(t)\big)^{p}\Gamma(t)^*
\Gamma'(t)\big].
\end{equation}
Furthermore, if the map $\S \ni t\mapsto \Gamma(t)-1 \in S_p$ is continuously differentiable in norm of $S_p$, then the statement already holds for $p$ instead of $p+1$.
\end{lemma}

\begin{proof}
For simplicity, let us set $A(t):=1-\Gamma(t)$ for any $t \in \S$ and recall from \cite[Eq.~XI.2.11]{GGK} that
$\det_{p+1}\big(\Gamma(t)\big) = \det\big(1+R_{p+1}(t)\big)$
with
\begin{equation*}
R_{p+1}(t):= \Gamma(t)\exp\Big\{\sum_{j=1}^p\frac{1}{j}A(t)^j\Big\} -1\ .
\end{equation*}
Then, for any $t,s \in \S$ with $s\neq t$ one has
\begin{eqnarray*}
\frac{\det_{p+1}\big(\Gamma(s)\big)}{\det_{p+1}
\big(\Gamma(t)\big)}
&=&\frac{\det\big(1+R_{p+1}(s)\big)}{\det
\big(1+R_{p+1}(t)\big)} \\
&=& \frac{\det\big[\big(1+R_{p+1}(t)\big)\big(1+B_{p+1}(t,s)\big)\big]}{\det
\big(1+R_{p+1}(t)\big)}\\
&=& \det\big(1+B_{p+1}(t,s)\big)
\end{eqnarray*}
with
$B_{p+1}(t,s) = \big(1+R_{p+1}(t)\big)^{-1} \big(R_{p+1}(s)-R_{p+1}(t)\big)$.
Note that $1+R_{p+1}(t)$ is invertible in $\B(\H)$ because $\det_{p+1}\big(\Gamma(t)\big)$ is non-zero.
With these information let us observe that
\begin{eqnarray}\label{presque2}
\frac{\frac{
\det_{p+1}(\Gamma(s))- \det_{p+1}(\Gamma(t))}{|s-t|}
}{\det_{p+1}(\Gamma(t))}
= \frac{1}{|s-t|}
\big[\det\big(1+B_{p+1}(t,s)\big)-1\big]\ .
\end{eqnarray}
Thus, the statement will be obtained if the limit $s\to t$ of this expression exists and if this limit is equal to the r.h.s. of \eqref{aobtenir}.

Now, by taking into account the asymptotic development of $\det(1+\e X)$ for $\e$ small enough, one obtains that
\begin{eqnarray}\label{presque}
\nonumber &&\lim_{s \to t} \frac{1}{|s-t|}
\big[\det\big(1+B_{p+1}(t,s)\big)-1\big] \\
\nonumber &=& \lim_{s \to t} \tr\bigg[\frac{B_{p+1}(t,s)}{|s-t|}\bigg] \\
&=&
\lim_{s \to t} \tr\bigg[H_{p+1}(t)^{-1}\frac{H_{p+1}(s)-H_{p+1}(t)}{|s-t|}\bigg]
\end{eqnarray}
with $H_{p+1}(t):=\big(1-A(t)\big)\exp\big\{\sum_{j=1}^p\frac{1}{j}A(t)^j\big\}$.
Furthermore, it is known that the function $h$ defined for $z \in \C$ by $h(z):=z^{-(p+1)}(1-z)\exp\big\{\sum_{j=1}^p \frac{1}{j}z^j\big\}$ is an entire function,
see for example \cite[Lem.~6.1]{S1}.
Thus, from the equality
\begin{equation}\label{bellavista}
H_{p+1}(t) = A(t)^{p+1}h\big(A(t)\big)
\end{equation}
and from the hypotheses on $A(t)\equiv 1-\Gamma(t)$ it follows that the map $\S \ni t \mapsto H_{p+1}(t)\in S_1$ is continuously differentiable in the norm of $S_1$. Thus, the limit \eqref{presque} exists, or equivalently the limit \eqref{presque2} also exists.
Then, an easy computation using the geometric series leads to the expected result, {\it i.e.}~the limit in \eqref{presque} is equal to the r.h.s.~of \eqref{aobtenir}.

Finally, for the last statement of the lemma, it is enough to observe from \eqref{bellavista} that the map $\S \ni t \mapsto H_{p}(t)\in S_1$ is continuously differentiable in the norm of $S_1$ if the map
$\S \ni t\mapsto \Gamma(t)-1 \in S_p$ is continuously differentiable in norm of $S_p$. Thus the entire proof holds already for $p$ instead of $p+1$.
\end{proof}

\begin{lemma}\label{l1}
Assume that the map $\S \ni t\mapsto \Gamma(t)-1 \in S_p$ is continuous in norm of $S_p$ and is continuously differentiable in norm of $\K(\H)$. Then for an arbitrary $q\geq p$ one has:
\begin{equation*}
\int_\S \tr\big[\big(1-\Gamma(t)\big)^{q}\Gamma(t)^*
\Gamma'(t)\big]\d t  = \int_\S \tr\big[\big(1-\Gamma(t)\big)^{p}\Gamma(t)^*
\Gamma'(t)\big]\d t\ .
\end{equation*}
\end{lemma}

\begin{proof}
One first observe that for $q>p$ one has
\begin{eqnarray*}
M_{q}(t)&:=&\tr\big[\big(1-\Gamma(t)\big)^{q}\Gamma(t)^*
\Gamma'(t)\big] \\
&=& \tr\Big[\big(1-\Gamma(t)\big)^{q-1}\Gamma(t)^*
\Gamma'(t)- \Gamma(t)\big(1-\Gamma(t)\big)^{q-1}\Gamma(t)^*
\Gamma'(t)\Big] \\
&=& M_{q-1}(t) -\tr\big[\big(1-\Gamma(t)\big)^{q-1}
\Gamma'(t)\big]
\end{eqnarray*}
where the unitarity of $\Gamma(t)$ has been used in the third  equality. Thus the statement will be proved by reiteration if one can show that
\begin{equation}\label{aintegrer}
\int_\S \tr\big[\big(1-\Gamma(t)\big)^{q-1}
\Gamma'(t)\big] \d t
\end{equation}
is equal to $0$.

For that purpose, let us set for simplicity $A(t):=1-\Gamma(t)$ and observe that for $t,s \in \S$ with $s\neq t$ one has
\begin{eqnarray*}
\tr[A(s)^q]-\tr[A(t)^q] =
\tr\big[A(s)^q-A(t)^q\big] =
\tr\Big[P_{q-1}\big(A(s), A(t)\big)\, \big(A(s)-A(t)\big) \Big]
\end{eqnarray*}
where $P_{q-1}\big(A(s), A(t)\big)$ is a polynomial of degree $q-1$ in the two non commutative variables $A(s)$ and $A(t)$. Note that we were able to use the cyclicity because on the assumptions $q-1\geq p$ and $A(t)\in S_p$ for all $t \in \S$. Now, let us observe that
\begin{eqnarray*}
&&\bigg|\frac{1}{|s-t|}
\tr\Big[P_{q-1}\big(A(s), A(t)\big)\, \big(A(s)-A(t)\big)\Big]- \tr\Big[ P_{q-1}\big(A(t),A(t)\big)\,A'(t)\Big]\bigg|
\\
&\leq& \Big\|\frac{A(s)-A(t)}{|s-t|}\Big\|  \,\Big|
\tr\big[P_{q-1}\big(A(s), A(t)\big)-P_{q-1}\big(A(t), A(t)\big)\big]
\Big| \\
&& + \Big\|\frac{A(s)-A(t)}{|s-t|}-A'(t)\Big\|
\,\Big|
\tr\big[P_{q-1}\big(A(t), A(t)\big)\big]\Big|\ .
\end{eqnarray*}
By assumptions, both terms vanish as $s \to t$. Furthermore, one observes that $P_{q-1}\big(A(t),A(t)\big)=q A(t)^{q-1}$. Collecting these expressions one has shown that:
\begin{equation*}
\lim_{s\to t}\frac{\tr[A(s)^q]-\tr[A(t)^q]}{|s-t|}- q\,\tr[A(t)^{q-1}A'(t)] =0\ ,
\end{equation*}
or in simpler terms $\frac{1}{q}\frac{\d\, \tr[A(t)^q]}{\d t} = \tr[A(t)^{q-1}A'(t)]$. By inserting this equality into \eqref{aintegrer} one directly obtains that this integral is equal to $0$, as expected.
\end{proof}

\end{document}